%% file: y2.tex
\documentclass{sig-alternate}

\usepackage[font=footnotesize,format=plain,labelfont=bf,up,textfont=it,up]{caption}
\usepackage{subfig}
\usepackage{epstopdf}
\usepackage{amssymb}
\usepackage{amsmath}
\usepackage{paralist} 
\renewcommand\subsubsection[1]{{\normalfont\normalsize\bfseries{#1.}}}
\usepackage{algorithm}
\usepackage[noend]{algorithmic}
\usepackage{times}
\usepackage{url}
\usepackage{xspace}
\usepackage{scalefnt}
\usepackage[amssymb]{SIunits} 

\usepackage{tikz}
\usetikzlibrary{arrows,shapes}

\newcounter{foo}
\newtheorem{theorem}[foo]{Theorem}

\newtheorem{lemma}[foo]{Lemma}

\def\calP{{\mathcal{P}}}

\newcommand{\mw}{\milli\watt\xspace}
\newcommand{\dBm}{\deci\bel\milli\xspace}

\newcommand{\prob}[1]{\textsc{#1}}
\newcommand{\capacity}{\prob{Link Capacity}}

\newcommand{\multicapacity}{\prob{Multi-Channel Link Capacity}}

\newcommand{\geomodel}{\textsc{geo-sinr}}
\newcommand{\decaymodel}{\textsc{mb-sinr}}

\newcommand{\PRR}{\text{PRR}\xspace}
\newcommand{\RSS}{\text{RSS}\xspace}

\newenvironment{myitemize}{
\begin{itemize}
  \setlength{\partopsep}{0pt}
  \setlength{\topsep}{0pt}
  \setlength{\itemsep}{0pt}
  \setlength{\parskip}{-0pt}
  \setlength{\parsep}{-0pt}
}{\end{itemize}}

\begin{document}

\title{Wireless Scheduling Algorithms\\ in Complex Environments
}

\numberofauthors{6}

\author{
%
%
\alignauthor
Helga Gudmundsdottir\\
       \affaddr{School of Computer Science, CRESS}\\
       \affaddr{Reykjavik University, Iceland}\\
       \email{helgag10@ru.is}
\alignauthor
Eyj\'olfur I.~\'Asgeirsson\\
       \affaddr{School of Science and Engineering, ICE-TCS}\\
       \affaddr{Reykjavik University, Iceland}\\
       \email{eyjo@ru.is}
\alignauthor 
Marijke H. L. Bodlaender\\
       \affaddr{School of Computer Science, ICE-TCS}\\
       \affaddr{Reykjavik University, Iceland}\\
       \email{marijke12@ru.is}
\and  
\alignauthor Joseph T. Foley\\
       \affaddr{School of Science and Engineering}\\
       \affaddr{Reykjavik University, Iceland}\\
       \email{foley@ru.is}
\alignauthor Magn\'us M.~Halld\'orsson\\
       \affaddr{School of Computer Science, ICE-TCS}\\
       \affaddr{Reykjavik University, Iceland}\\
       \email{mmh@ru.is}
\alignauthor Ymir Vigfusson\\
       \affaddr{School of Computer Science, ICE-TCS, CRESS}\\
       \affaddr{Reykjavik University, Iceland}\\
       \email{ymir@ru.is}
}



\date{\today}
\maketitle

\begin{abstract}
Efficient spectrum use in wireless sensor networks 
through spatial reuse requires effective models of packet reception
at the physical layer in the presence of interference.  
Despite recent progress in analytic and simulations research into
worst-case behavior from interference effects, these efforts
generally assume geometric path loss and isotropic 
transmission, assumptions which have not been borne out in experiments.


Our paper aims to provide a methodology for grounding 
theoretical results into wireless interference in experimental reality. 
We develop a new framework for wireless algorithms in which
distance-based path loss is replaced by an arbitrary gain matrix,
typically obtained by measurements of received signal strength (RSS).
Gain matrices allow for the modeling of complex environments, e.g., with
obstacles and walls. 
We experimentally evaluate the framework in two indoors testbeds with
20 and 60~motes, and confirm superior predictive performance in packet reception rate for a gain matrix model
over a geometric distance-based model.

At the heart of our approach is a new parameter $\zeta$ called \emph{metricity} which indicates how
close the gain matrix is to a distance metric, effectively measuring
the complexity of the environment. A powerful theoretical feature of this parameter is
that \emph{all} known SINR
scheduling algorithms that work in general metric spaces carry over to arbitrary gain matrices 
and achieve equivalent performance guarantees in terms of $\zeta$ as previously obtained in terms of the path loss constant.
Our experiments confirm the sensitivity of $\zeta$ to the nature of
the environment.  
Finally, we show analytically and empirically how multiple channels can be leveraged to improve metricity and thereby performance.
%
We believe our contributions will facilitate experimental validation for recent advances in algorithms for physical wireless
interference models.

\end{abstract}


\section{Introduction}

There is mounting demand for tomorrow's wireless networks to provide
higher performance while lowering costs.
A central challenge in meeting this demand is to improve the utilization of the wireless spectrum to enable 
simultaneous communications at the same radio frequency. 
To accommodate research into efficient use of wireless channels at large-scale, for instance through spatial reuse,
we require practical models of signal propagation behavior and reception at the physical layer
in the presence of wireless interference.

%


Early models of worst-case wireless communication under interference were
graph-based, most commonly based on distances.  In comparison,
physical models, or \emph{SINR} (signal to interference and noise
ratio) models, capture two important features of reality: signal
strength decays as it travels (rather than being a binary property)
and interference accumulates (rather than being a pairwise relation).

%

Analytic work on SINR -- introduced by Gupta and Kumar \cite{kumar00}
in an average-case setting and Moscibroda and Wattenhofer \cite{MoWa06} in worst-case
-- has generally assumed \emph{geometric path loss},
referred to here as the {\geomodel} model:
signals decay as a fixed polynomial of the distance traveled.

While free space exhibits geometric decay, the reality for \emph{real-world} wireless environments is more complex.
When located above an empty plane, a signal bounces off the ground, resulting in 
complicated patterns of superpositions known as \emph{multi-path fading}.
Most real scenarios are more complex, with walls and obstructions.
In particular, cityscape and indoor environments are notoriously
hard to model.
%
Moreover, the simple range-based models often make further assumptions into geometric path loss that do not concord with experiments, 
such as smooth and isotropic polynomial decrease in the signal strength.
In fact, quoting recent meta-analysis \cite{baccour2012radio}, ``link quality is not
correlated with distance.'' 

Various stochastic extensions of geometric path loss have been proposed to
address the observed variability in signal propagation.
The most common are \emph{log-normal shadowing} and \emph{Rayleigh fading}
for addressing long- and short-distance variability, respectively.
Both modify the signal strength multiplicatively by an exponentially
distributed random variable.
These models are highly useful both for generating input for signal propagation 
simulations and for average-case analysis of wireless interference algorithms.

A complementary view to stochastic studies,
%
with deep roots in computer science theory, is
to allow for worst-case behavior and obtain guarantees that hold for
\emph{all} instances to the problem at hand.  To avoid such results becoming too pessimistic,
proper characterizations or parameterizations are often essential.
Our goal is to contribute to such ``any-case'' analysis that avoids
making assumptions about the environment that may not be reflected in actual real-world scenarios.

\subsubsection{Our contributions}
We propose moving theoretical algorithm design away from assuming geometric path loss 
models to an abstract SINR formulation with a matrix
representing the fading (or signal decay) between pairs of nodes
in an arbitrary environment.
The matrix would typically be generated from direct measurements of
\emph{received signal strength} (RSS) provided by motes, as proposed
by experimentalists \cite{son2006,reis2006,MaheshwariJD2008}.
The RSS matrix could also be generated by other means, such as
by inference, history, stochastic models or by accurate environmental models.

Following this approach, worst-case algorithmic analysis is heavily contingent on the
contents of the RSS matrix, with unconstrained settings causing computational intractability.
We introduce a new measure that reflects the attenuation complexity of the
environment described by the RSS matrix.  
Dubbed \emph{metricity} and denoted $\zeta$, this parameter
intuitively represents how close the RSS matrix is to a distance
metric.
From a theoretical standpoint, the definition of metricity has extensive implications:
\textbf{All} SINR algorithms that work in arbitrary metric spaces work 
seamlessly in the abstract model, with performance ratio in terms of metricity that is equivalent to the original dependence on 
the path loss constant. 


In an experimental evaluation on 
two testbeds of 20 and 60 nodes, our measurements indicate
that the metricity parameter corresponds to the complexity of the
environment. The experiments also suggest that the SINR model --
without the geometric assumption -- is of high fidelity, capturing
signal propagation and reception well, even in environments with
obstacles and lack of line-of-sight.

We further address the effect of multi-path fading by giving
transmitters the choice of several channels/frequencies. Empirically, we find that in an
environment with extensive multi-path propagation (but otherwise
simple), the choice improves the metricity parameter significantly.
Analytically, we show that a known algorithm for capacity
maximization can be extended to handle multiple channels without loss in performance.

\subsubsection{Roadmap}
In the following section, we formally define our concepts,
describe and calibrate our experimental setup in 
and validate the basic premises of our framework.
We analyze the metricity parameter $\zeta$ in Section~\ref{sec:metricity} and present experimental results.
By leveraging the metricity concept in our framework, 
we introduce an approach for tackling multi-path fading using multiple frequencies in Section 
\ref{sec:multichannel} and present experimental and theoretical results.
We survey related work in Section~\ref{sec:related} and
conclude in Section~\ref{sec:discussion}.




\input{experiments}

\input{metricity}

\input{multichanneltheory}


\input{related}

\section{Conclusion}
\label{sec:discussion}\label{sec:conclusion}

Effective use of the wireless spectrum requires an understanding of interference from theoretical and experimental vantage points.
A growing body of algorithmic work on worst-case wireless interference under the SINR threshold model assumes that signals decay geometrically with distance,
the \geomodel{} model.

We outline an approach for incorporating realism into the interference model while seamlessly generalizing previous theoretical results.
By leveraging a matrix of pairwise RSS between wireless motes instead of geometric path loss, the \decaymodel{} model was
significantly better at predicting PRR performance in our experiments on two indoors testbeds.
The RSS matrix also appears resilient to temporal factors, with prediction accuracy of 95\% when using an RSS matrix created weeks in advance.

We defined the notion of metricity, which quantifies the proximity of the RSS matrix to a distance metric.
Through experiments, we showed how the concept effectively measures the complexity of the underlying environment.
With metricity as a harness, worst-case theoretical results that hold under general metrics in the \geomodel{} model 
can now be translated with trivial modifications to the more realistic \decaymodel{} model.
Moreover, the translation retains almost identical performance ratios for all such algorithms, with the metricity parameter $\zeta$ replacing the path loss constant $\alpha$ in the \geomodel{} model.

As a case study of the metricity concept, we consider multi-path fading by allowing transmitters to choose between several wireless frequencies.
When using multi-path fading in our experiments, we found that environments with extensive multi-path propagation 
exhibited better metricity values.
By fusing empirical measurements into an analytical model commonly used for worst-case theoretical analysis, 
our approach illustrates a methodology for harmonizing algorithmic theory of wireless interference with real-world observations.

\bibliographystyle{abbrv}
\bibliography{references}

\end{document}

%% file: experiments.tex
\newcommand{\classroom}{\textsc{Tb-20}}
\newcommand{\basement}{\textsc{Tb-60}}
%
%

\section{Model Validation}\label{sec:experiments}

Our first order of action is to verify that the SINR model, without
the geometric assumption, is faithful to reality.  We assess the
predictability of packet reception rate (\PRR) under interference, and
the assumption of the additivity of interference, by comparing our
abstract model to the original \geomodel{} through experiments. The 
experiments are conducted in two testbeds (Fig~\ref{fig_testbed}): one
in the middle of a large open classroom (\classroom{}) and another in
a challenging basement corridor (\basement{}).  

\subsection{The Physical Model}
\label{sec:models}

The SINR model is based on two key principles: \textbf{(i)} a signal decays as
it travels from a sender to a receiver, and \textbf{(ii)} interference
-- signals from other sources than the intended transmitter -- accumulates.
A transmission is successfully received if and only if the strength of the received signal relative
to interference is above a given threshold.

Formally, a \emph{link} $\ell_v = (s_v, r_v)$ is given by a pair of
nodes, sender $s_v$ and a receiver $r_v$.  The channel \emph{gain}
$G_{uv}$ denotes the reciprocal of the signal decay of $\ell_u$ as received at
$r_v$.  If a set $S$ of links transmits simultaneously,
then the SINR at $\ell_v$ is
\begin{equation}
 \text{SINR}_v := \frac{P_v G_{vv}}{N + \sum_{u \in S} P_v G_{uv}}\ ,
\label{eqn:sinr}
\end{equation}
where $P_v$ is the power used by the sender $s_v$ of $\ell_v$, and $N$
is the ambient noise.  In the \emph{thresholded} SINR model, the
transmission of $\ell_v$ is \emph{successful} iff $\text{SINR}_v \ge
\beta$, where $\beta$ is a hardware-dependent constant.

The common assumption of \emph{geometric path loss} in SINR models states
that the gain is inversely proportional to a fixed polynomial of the
distance traveled, \emph{i.e.}, 
$G_{uv} = d(s_u, r_v)^{-\alpha}$, where the range of the \emph{path loss constant}
$\alpha$ is normally between 1 and 6. The geometric path loss assumption
is valid in free space; we have $\alpha=2$ in perfect vacuum.

The {\decaymodel} model refers to the SINR formula (\ref{eqn:sinr}) applied to a
general \emph{gain matrix} $G$ obtained through pairwise RSS measurements.

\subsection{Experimental Setup}

\begin{figure}[bh!]
\centering
\subfloat{
\includegraphics[width=0.48\columnwidth]{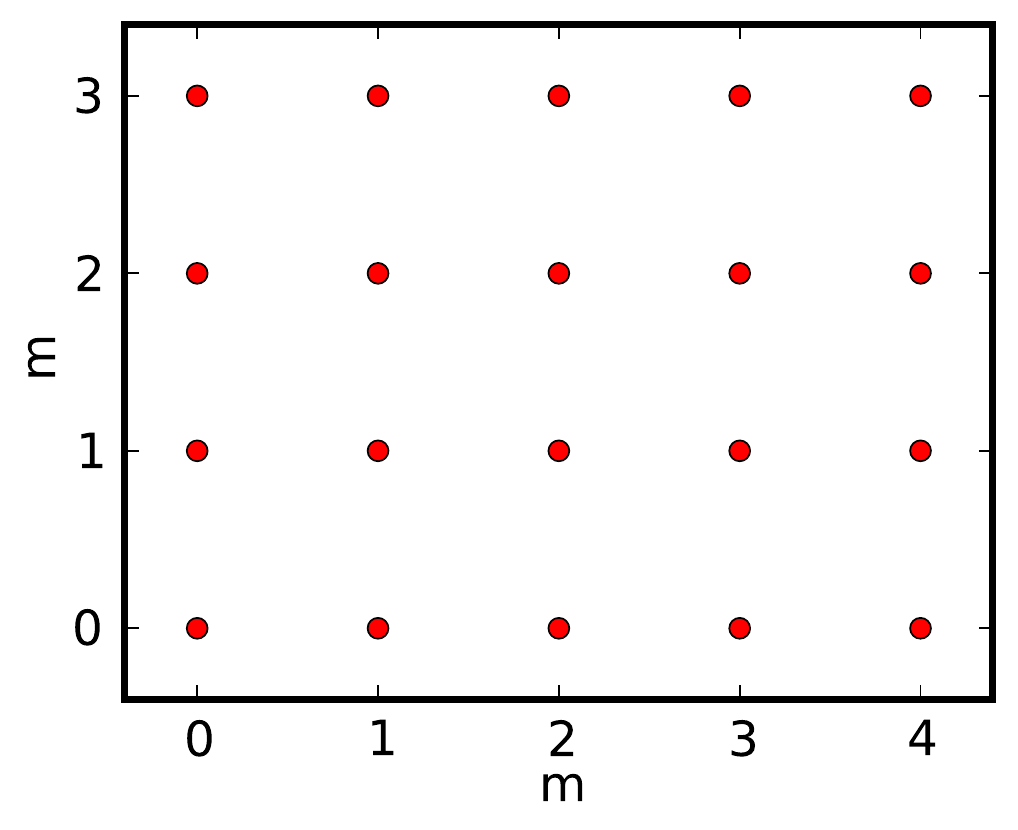}
}
\subfloat{
\includegraphics[width=0.48\columnwidth]{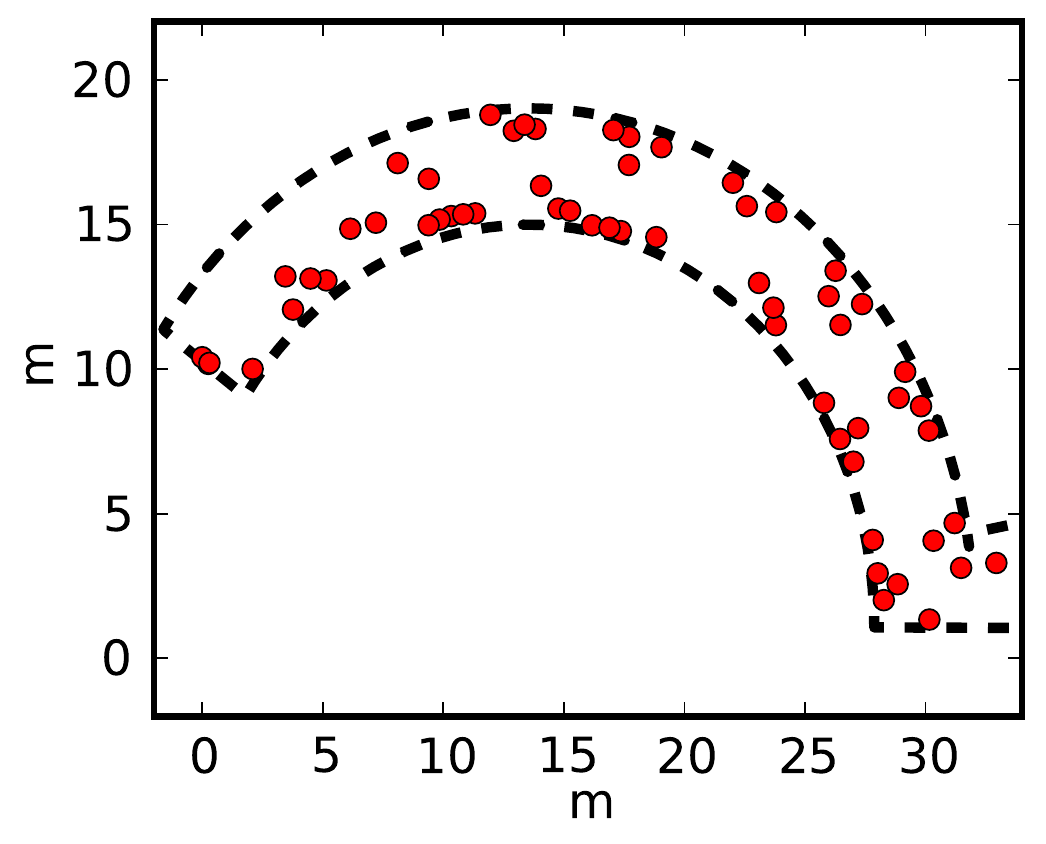}
}
\caption{Topologies of our 20-node testbed (\classroom{}) \textbf{(left)} and 60-node curved corridor testbed (\basement{}) \textbf{(right)}.} 
\label{fig_testbed}
\end{figure}

\subsubsection{Wireless hardware}
Since a motivating goal of our study is to understand raw interference between wireless transceivers, 
we elected to operate at the physical-layer of a wireless device.
We needed a mote with granular control over MAC-level capabilities, such as power and frequency control,
over one tailored to specific protocol stacks, such as the 802.11 suite.
For example, we require the ability to disable low-level features such as clear channel assessment (CCA).

We chose the Pololu Wixel, a development board
for the TI CC2511F32 \cite{ti_cc2511f32_datasheet}, as our mote hardware platform.
The CC2511F32 is an 8051 micro controller SystemOnChip with integrated 2.4~\giga\hertz{}
FM-transceiver stage (CC2500).
In addition to meeting our functional requirements, the Wixels are
inexpensive (14--20 USD), enabling larger scale deployments.



\begin{figure}[t!]
\centering%
\subfloat[\RSS matrix.]{%
\includegraphics[width=0.48\columnwidth]{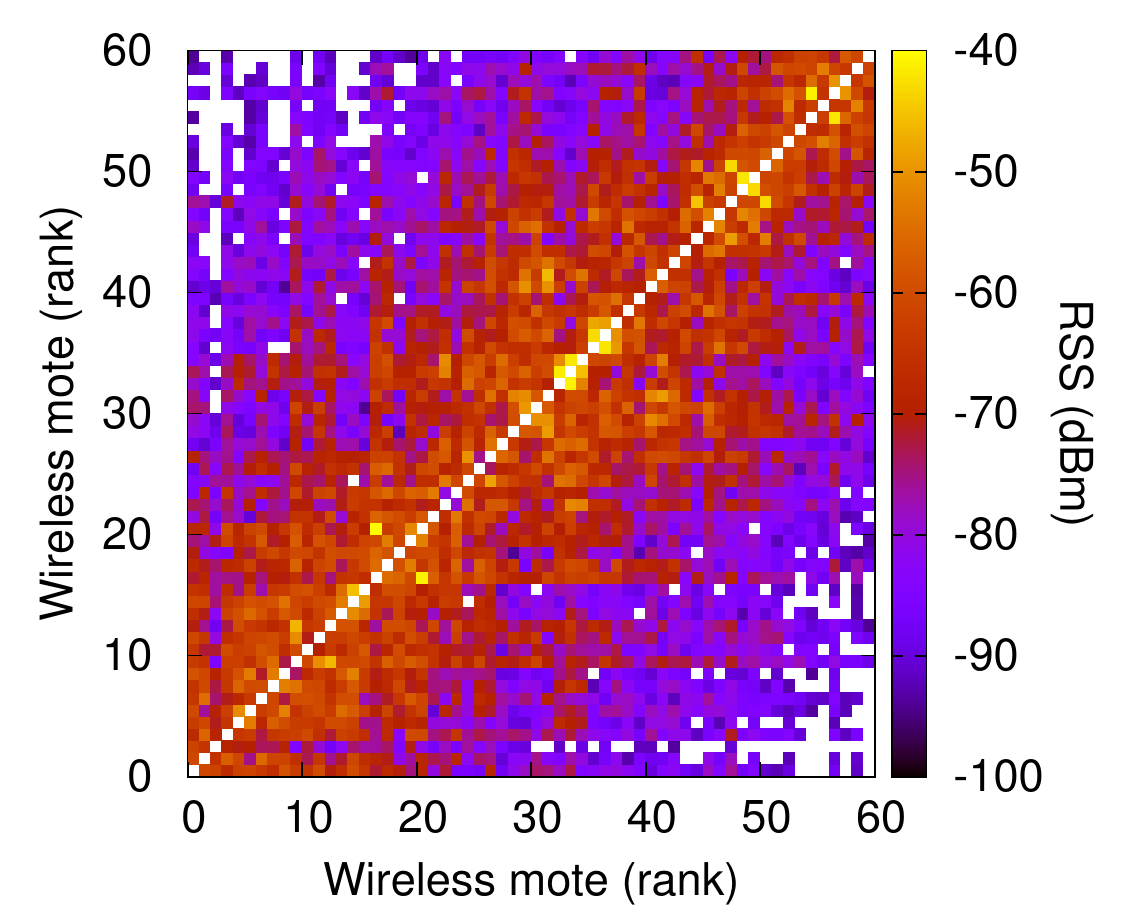}%
\label{fig_gains}%
}%
\subfloat[Comparing \RSS and geometric path loss ($d^{-2.18}$).]{
\includegraphics[width=0.48\columnwidth]{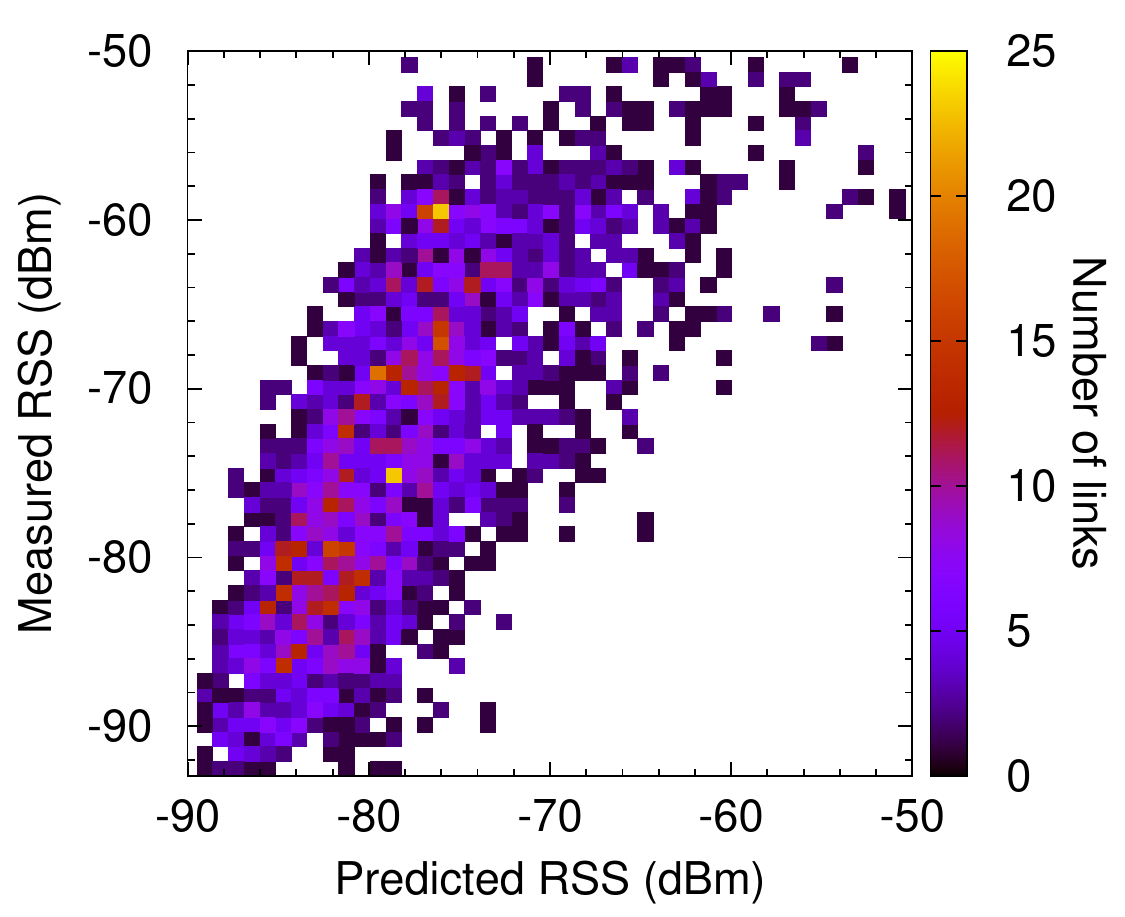}%
\label{fig_correlation}%
}
\label{fig_alphagains}
\caption{
\textbf{(a)} \textbf{\RSS matrix.} Gain between directed pairs of nodes in \basement{}, measured by RSS and averaged over 1000 packets.
\textbf{(b)} \textbf{Comparing measured and predicted \RSS.} Correlation between \RSS as predicted by distance with geometric path loss ($\alpha=2.18$) and measured in \basement{} testbed.
}
\end{figure}

\subsubsection{Configuration}
In our experiments, every sender node in each trial transmits a burst of 1000 packets
with 4~\milli\second{} delay between consecutive packets to facilitate successful delivery to the receiver.
The length of each packet is 22 bytes, including a 16-bit CRC.
Only packets that pass a CRC check are considered successful transmissions, 
with all error correction capabilities on the mote disabled.
The radio is configured to use data whitening and Minimum Shift Keying (MSK) modulation format.
During experiments, the wireless motes report details about packets
sent or received to an auxiliary log via USB which also provides
control signals and power for the experiments.
Packet details include the received signal strength (RSS) as an integer in \dBm.


\subsubsection{\classroom{} testbed} In the first testbed, we arranged 20 wireless motes on an $4$
$\times$ $5$ grid with 1~\meter{} spacing in an empty classroom. 
The motes in the grid were mounted on wooden poles 1~\meter{} from the ground
in order to minimize
reflection and attenuation from the ground; see Fig.~\ref{fig_testbed} for the topology.
The testbed was deployed temporarily for a focused set of tests.

\subsubsection{\basement{} testbed} In the second testbed, we suspended 60
wireless motes about 0.3~\meter{} from metal wire trays and 2.5~\meter{} from the concrete floor in a curved basement corridor;
see Fig.~\ref{fig_testbed} for the topology.
The corridor provides a challenging environment:
limited line of sight between motes, obstacles such as water pipes and thick
electric cables, and reinforced concrete walls.
Approximately 94\% of the directed links are in range for communication.
The length of the corridor is 40.1~\meter{}, 
the longest distance (direct line) between any two motes is 21.8~\meter{} while the shortest distance is 0.4~\meter{}.
The \basement{} testbed is a more permanent setup, with the experiments conducted over the span of several weeks.

\subsection{Model calibration}
We ran several experiments on the testbeds to gather calibration data for the \decaymodel{} and \geomodel{} models.
The figures with error bars show the median, and upper and lower quartiles of the distribution of experimental trials.

\subsubsection{Ambient noise parameter $\boldsymbol{N}$}
We evaluated $N$ by sampling the noise level registered by each mote, 
over several hours in both early morning and during nighttime.
All of our experiments use the 2.44~\giga\hertz{} frequency unless otherwise stated.
We found the average ambient noise in \basement{} to be around -99.1~\dBm, but considerably
higher in \classroom{} at around -94.4~\dBm,
in part due to external interference from 802.11 infrastructure. 


\subsubsection{Power setting $\boldsymbol{P}$}
We configured the wireless mote's transmission power to 1~\mw in all of our experiments,
and normalize the $P$ parameter as 1.

\subsubsection{\RSS matrix} 
We measured the RSS for all directed node pairs $(s_v,r_v)$ in both testbeds.
In each time slot, a chosen node transmits 1000 packets in a sequence, while
other nodes act as receivers. 
The procedure was repeated for each pair of directed nodes.
For temporal robustness, including day and night variations, the experiments were repeated at 
different times of the day.

Fig.~\ref{fig_gains} illustrates the $\RSS_{vv}$ for all node pairs $(s_v,r_v)$ in testbed \basement{}.
The motes on both axes are ordered by the angle of their polar coordinates due to the arced positioning of the corridor,
thus making neighboring motes in the testbed likely to be adjacent in Fig.~\ref{fig_gains}.
The figure further demonstrates that every mote can hear some other mote in the testbed, 
that some mote pairs cannot communicate, and that transmissions are not fully symmetric.

\subsubsection{Path loss constant $\boldsymbol{\alpha}$}
Assuming geometric signal decay, 
the best linear least-squares fit for the path loss constant $\alpha$ given link lengths and the RSS values from Fig.~\ref{fig_gains} was $\alpha = 2.18 \pm 0.07$.
%
By using geometric path loss $d(s_v,r_v)^{-2.18}$ 
to predict RSS values, we plot the correlation between the predicted RSS and measured RSS values in Fig.~\ref{fig_correlation}.
If the prediction were perfect, the points would fall on the $y=x$ diagonal line.
Instead, the results in Fig.~\ref{fig_correlation} confirm that 
geometric path loss is not a reliable predictor for RSS.

\begin{figure}[t]
\centering%
\subfloat[\decaymodel{}]{%
\includegraphics[width=0.50\columnwidth]{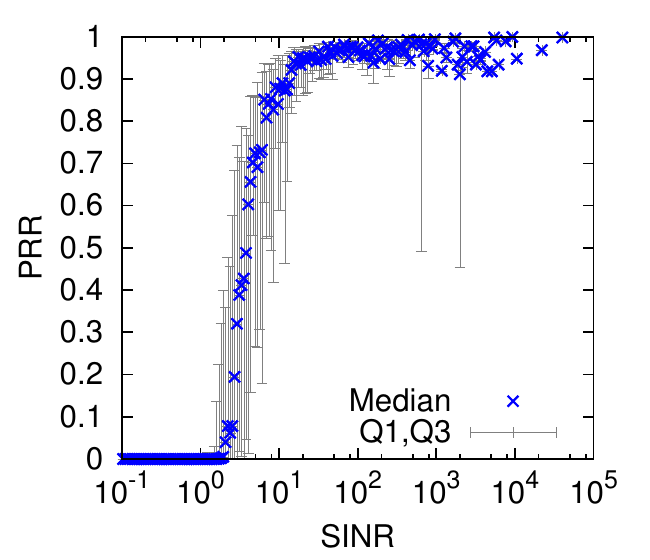}%
\label{fig_decay_sinr}%
}%
\subfloat[\geomodel{}]{%
\includegraphics[width=0.50\columnwidth]{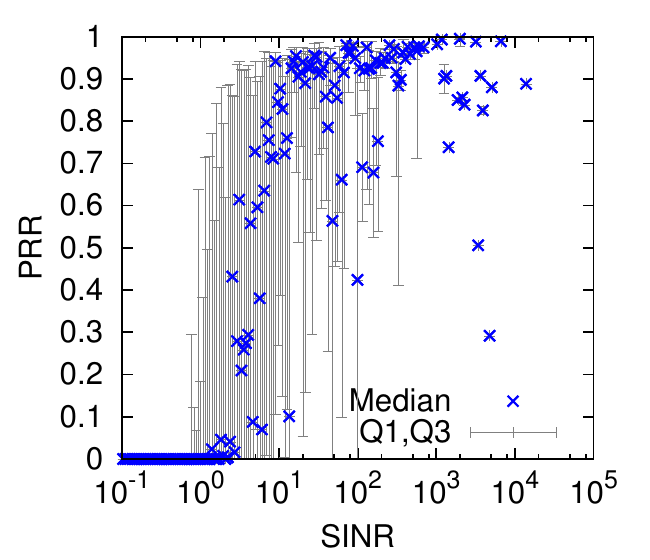}%
\label{fig_geo_sinr}%
}%
\label{fig_sinrmodels}
\caption{
\textbf{Packet reception rate for (a) \decaymodel{} (b) \geomodel{}.}
Fraction of packets correctly decoded by receivers in \basement{} as the SINR 
is varied by evaluating different pairs, possibly invoking multiple senders.
The plots show a transition from 0 to about 100\% PRR as the SINR grows.
}
\end{figure}

\subsection{Comparison of {\textsc{\scalefont{1.3}geo-sinr}} and \textsc{\scalefont{1.3}mb-sinr}}
We next evaluate the predictive power of the two models in experiments with varying interference.

\subsubsection{Controlling wireless interference in practice}
One of the challenges with hardware experiments is the synchronization of the wireless motes. 
Our focus on measuring interference
requires us to ensure that interfering motes are transmitting at the same time as the sender.
Although we investigate sets of links that are transmitting at the same time, 
our analysis is focused on the performance of individual links.
We therefore circumvent the problem of synchronizing the motes by running our experiments for each individual link in the link set.
To analyze how a single link $\ell_v$ would perform in the presence of the other sender-receiver pairs in a set $S$,
we make the links in $S \setminus \{\ell_v\}$ transmit continuously while we measure the transmission of $\ell_v$.
The continuous transmission by other senders ensures that the receiver in link $\ell_v$ experiences interference from other links.

\subsubsection{Experimental design}
With the synchronization issues in mind, we devised an experiment to compare
the predictive power of the two models.
We repeatedly select a random pair of nodes to act as sender and receiver, and
a subset of 1--10 other nodes to cause interference.
During the trial, the interfering nodes continuously transmit packets on the same frequency.
We deploy low-level packet filtering at the receiver to minimize processing overhead due to interfering packets.

\begin{figure}[t]
\centering%
\subfloat[Additivity in \decaymodel{}]{%
\includegraphics[width=0.50\columnwidth]{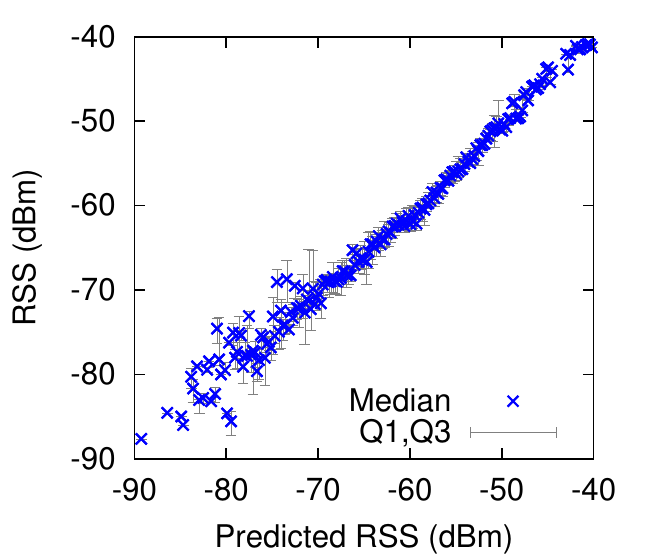}%
\label{fig_add_decay}%
}%
\subfloat[Additivity in \geomodel{}]{%
\includegraphics[width=0.50\columnwidth]{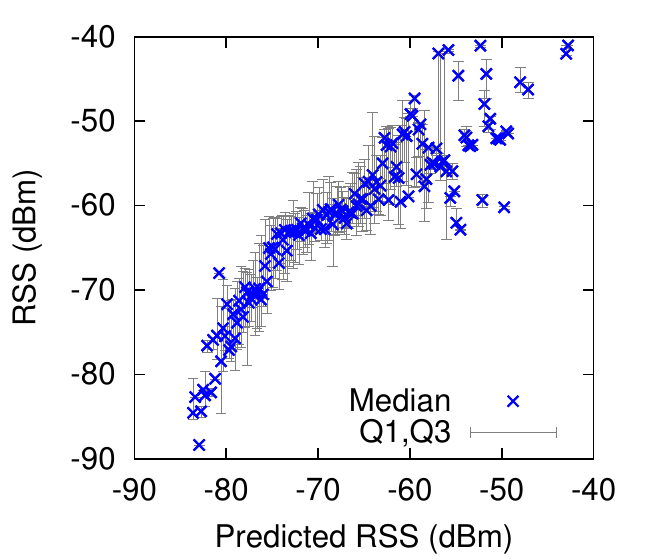}%
\label{fig_add_geo}%
}%
\label{fig_additivity}
\caption{
\textbf{Additivity in (a) \decaymodel{} and (b) \geomodel{}.} Correlation between predicted and measured (RSS) in \basement{}.
}
\end{figure}

\subsubsection{Packet reception rate by model}
%
For every possible link in each testbed, 
we generated over 15,000 packet transmission trials over the link and measured the packet reception rate as a function of the SINR as calculated by the two models.
Fig.~\ref{fig_decay_sinr} shows the \PRR as a function of the SINR in \decaymodel{}, calculated using the RSS matrix shown in Fig.~\ref{fig_gains}.
Corresponding
results for \geomodel{} where the SINR is based on distances between nodes are shown in Fig.~\ref{fig_geo_sinr}.

The \decaymodel{} behaves as expected:
generally the \PRR and SINR values are either both small or both large.
There is a swift transition from low to high \PRR as the SINR value increases. 
We note that occasional trials produce a small \PRR{} value despite SINR being large,
as indicated by the two large error bars where SINR $\approx 10^3$.
These outliers stem from occasionally no packets being received even for a large SINR, likely 
due to details of the testbed topology,
such as destructive interference caused by signal reflection.
Contrary to the \geomodel{} model, the \decaymodel{} model has a discernible threshold for successful transmissions. 


\subsubsection{Additivity of interference}
Among the assumptions made by the SINR model is that interference is additive.
In other words, if multiple senders transmit simultaneously, the RSS at the receiver can be estimated as the sum of the individual signals.
Fig.~\ref{fig_add_decay} and \ref{fig_add_geo} show the actual
RSS as a function of the predicted RSS as given by \decaymodel{} (Fig.~\ref{fig_add_decay})
and \geomodel{} (Fig.~\ref{fig_add_geo}).  
We note that the variability evident in the measured RSS arises due to sparsity of data in those regions.
If the additivity assumption is true, we would expect the values in the figures to fall on the diagonal line $y = x$.  

The \geomodel{} appears more closely described by a pair of line segments with
different slopes than a linear fit. 
Using linear regression, the coefficient of variation between the axes is low ($r^2 \simeq 0.031$), implying a low goodness-of-fit.
Conversely, the \decaymodel{} model has a strong linear trend,
with linear regression to the diagonal line incurring only $3.2\%$ error and producing a large coefficient of variation ($r^2 \simeq 0.968$)
between the predicted and measured RSS.
The \decaymodel{} therefore more closely captures the additivity of interference than the canonical \geomodel{} model.

\subsubsection{Sensitivity and specificity analysis}
The \geomodel{} and \decaymodel{} models can be viewed as binary classifiers that compare the SINR to a threshold ($\beta$)
to determine whether or not a transmission will be successful.
We say that a transmission is experimentally successful if $\PRR \geq T_{high}$ and declare it to be a failure if $\PRR \leq T_{low}$.
We focus on those links that were clearly either feasible or infeasible in our experiments, and set $T_{high} = 0.8$ and $T_{low} = 0.2$.
Roughly 6\% of the tested links fall within the $0.2-0.8$ range and are thus not considered.
A single instance in an SINR binary classifier can have four outcomes:
\newcommand{\TP}{\text{TP}}
\newcommand{\FP}{\text{FP}}
\newcommand{\TN}{\text{TN}}
\newcommand{\FN}{\text{FN}}
\begin{myitemize}
\item True positive (\TP): $\text{SINR} \geq \beta$, $\PRR \geq T_{high}$ 
\item False positive (\FP): $\text{SINR} < \beta$, $\PRR \geq T_{high}$ 
\item True negative (\TN): $\text{SINR} < \beta$, $\PRR < T_{low}$ 
\item False negative (\FN): $\text{SINR} \geq \beta$, $\PRR < T_{low}$ 
\end{myitemize}
A binary classifier incurs an inherent trade-off between true positive rate (sensitivity), defined as $\frac{\TP}{\TP + \FN}$, and false positive rate (1-specificity), defined as $1-\frac{\TN}{\FP+\TN} = \frac{\FP}{\FP+\TN}$. 
The trade-off balance can normally be tuned by a threshold parameter of the classifier, in this case $\beta$.
By varying $\beta$, the trade-off can be graphically depicted on a ROC-curve (Receiver Operating Characteristic) that shows true and false positive rates on two axes.
If $\beta = 0$, the classifier predicts that all transmissions will be successful,
and if $\beta$ is large, the classifier predicts that all transmissions will fail.
A na\"ive classifier making uniformly random guesses would fall on the diagonal line from $(0,0)$ to $(1,1)$, whereas the $(0,1)$ point denotes perfect classification.

\begin{figure}[t]
\centering%
\subfloat[ROC-curve for \basement{}]{%
\includegraphics[width=0.5\columnwidth]{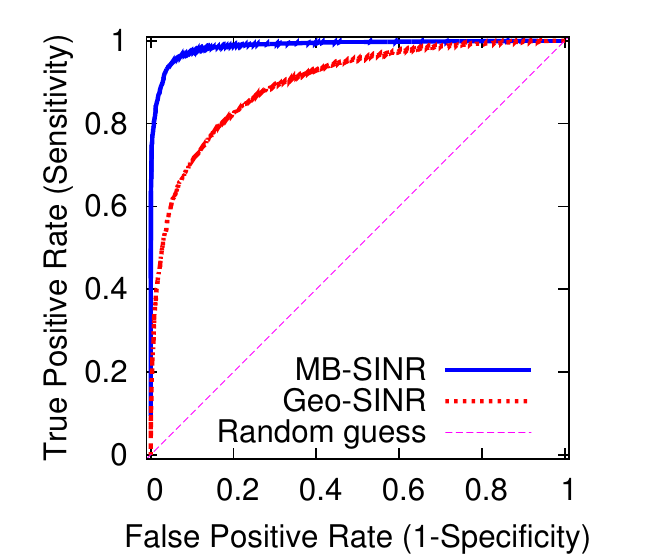}%
\label{fig_roc_basement}%
}%
\subfloat[ROC-curve for \classroom{}]{%
\includegraphics[width=0.5\columnwidth]{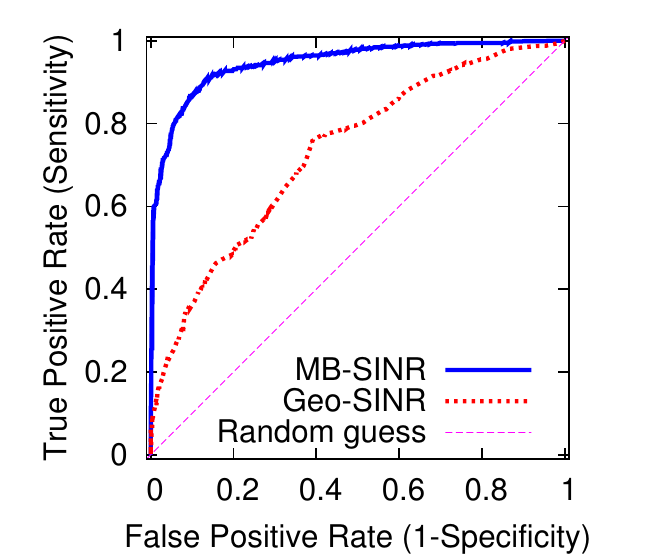}%
\label{fig_roc_m110}%
}%
\label{fig_rocs}%
\caption{
\textbf{ROC-curves for (a) \basement{} and (b) \classroom{}.} Comparison of \decaymodel{} and \geomodel{} as estimators
for successful transmission of packets as the acceptance threshold $\beta$ is varied.
Each trial consists of 1000 packets exchanged in \basement{}. A positive trial outcome has $\PRR \geq80\%$ whereas a
negative one has $\PRR \leq20\%$.
}
\end{figure} 

Fig.~\ref{fig_roc_basement} and \ref{fig_roc_m110} show the ROC-curves for the \basement{} and \classroom{} testbeds, respectively.
\decaymodel{} provides significantly better classification than \geomodel{}.
In \basement{}, the best trade-off between true and false positive rates occurs when $\beta = 2.15$,
with a true positive rate of 94.8\% and false positive rate of 5.2\%.

In contrast, the predictions made by the canonical \geomodel{} on the same testbed plateaus at true positive rates of 81.4\% and false positive rates of 18.6\%.
Both models give less accurate predictions for \classroom{} compared to \basement{}.
The topology for \classroom{} is more compact than \basement{}, with fewer number of trials.
The larger and more variable ambient noise in \classroom{} makes it more difficult to accurately predict the outcome of a transmission
than in the \basement{} testbed.

As expected, \decaymodel{} provides significantly superior predictive power for PRR than \geomodel{} on both testbeds.
Our results are robust against modifying the thresholds to $T_{high},T_{low} = 0.5$. 
Note that \RSS measurements used to compute SINR in \decaymodel{} were performed weeks in advance of these experiments.
This suggest that the \RSS matrix
is resilient to temporal factors, with \decaymodel{} correctly predicting nearly 95\% of all instances.

%% file: metricity.tex
\begin{figure*}[t]
\centering%
\subfloat[$\zeta_v$ in \classroom{} and \basement{}]{%
\includegraphics[width=0.5\columnwidth]{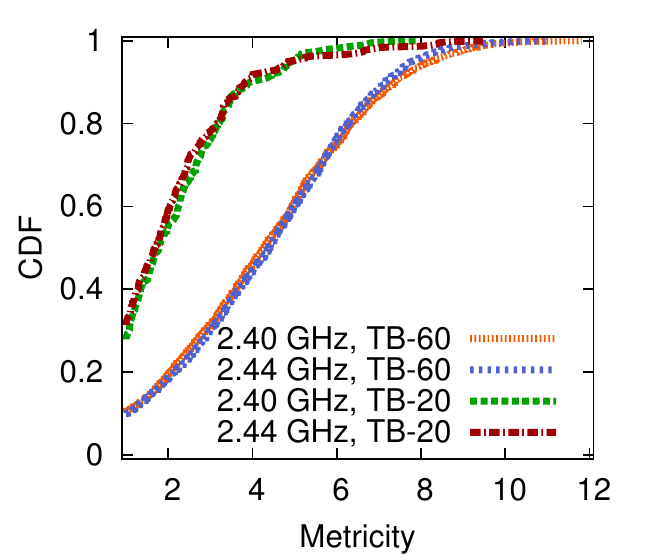}%
\label{fig_zeta}%
}%
\subfloat[$\zeta_v$ of different subsets of motes in \basement{}]{%
\includegraphics[width=0.5\columnwidth]{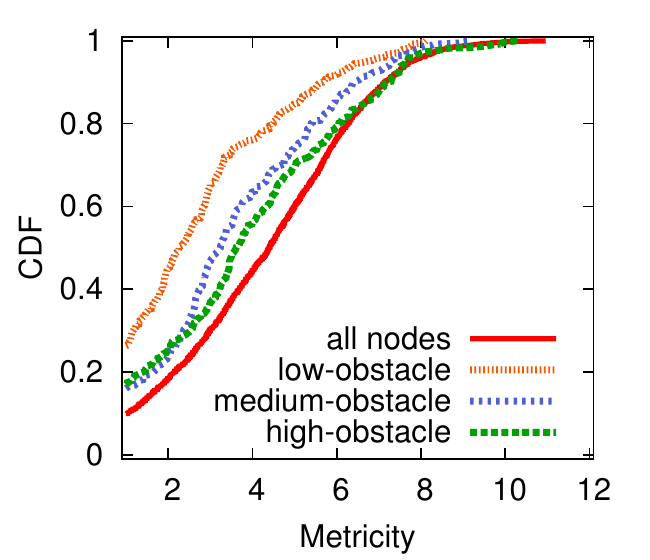}%
\label{fig_zeta_groups}%
}%
\subfloat[$\zeta$ of subsets of motes with increasing number of nodes]{%
\includegraphics[width=0.5\columnwidth]{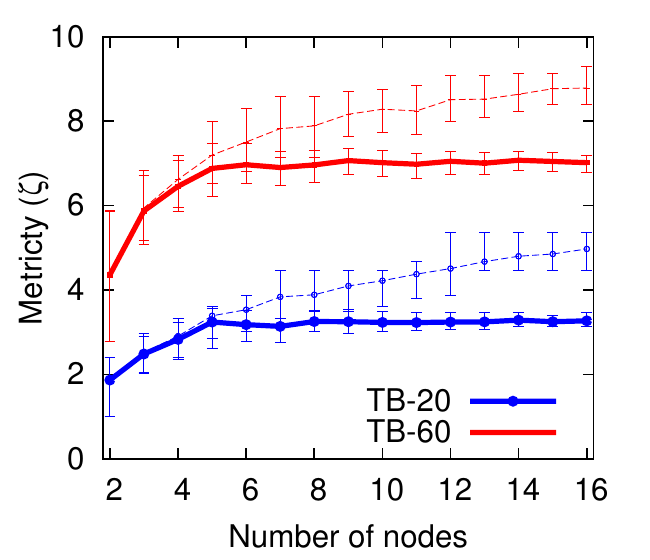}%
\label{fig_zeta_size}%
}%
\subfloat[$\zeta$ of subsets with increasing distance between nodes]{%
\includegraphics[width=0.5\columnwidth]{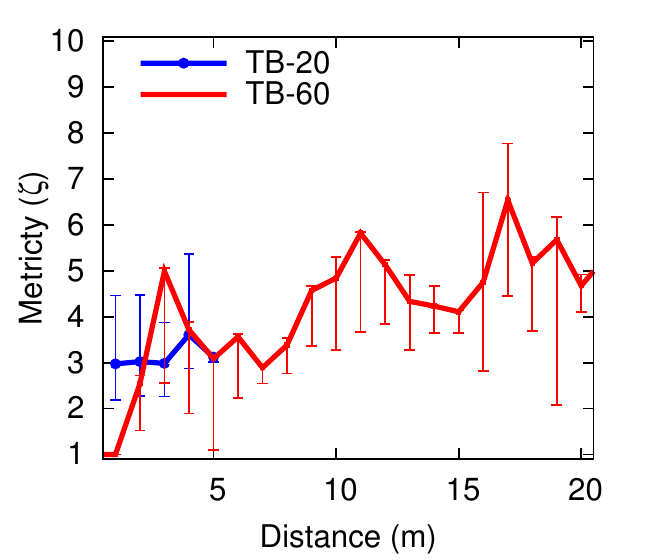}%
\label{fig_zeta_distance}%
}%
\label{fig_zetas}%
\caption{
\textbf{(a)} \textbf{Metricity ($\boldsymbol\zeta_v$)}.
CDF comparison of computed $\zeta_v$ values for RSS matrices on both testbeds.
The lines show $\zeta_v$ values for two different frequencies, which are representative of other frequencies in the respective testbeds.
\textbf{(b)} \textbf{Environmental factors}.
CDF comparison of computed $\zeta_v$ values on three different subsets of nodes. \textit{low-obstacle} has the fewest environmental obstacles; \textit{high-obstacle} has the most.
\textbf{(c)} \textbf{Impact of scale on $\boldsymbol\zeta$}.
 Metricity compared to the number of nodes calculated as the average metricity of 200 randomly generated subsets of nodes of size at most 16, calculated for both testbeds. 
 The thick lines represent the average of the $95^{\text{th}}$ percentile $\zeta$ of subsets of the same size, whereas the thin lines represent the average maximum $\zeta$ of the subsets of the same size.
 The error bars represent the inter-quartile range.
\textbf{(d)} \textbf{Influence of pairwise distances on $\boldsymbol\zeta$}.
 Metricity calculated on subsets of nodes with similar distances on both testbeds.
 The lines represent the $95^{\text{th}}$ percentile of the $\zeta$ value of each set.
 The error bars represent the interval between the highest value of $\zeta$ and the median.
}
\end{figure*}

\section{Metricity}\label{sec:metricity}

A plethora of important wireless interference algorithms rely on the {\geomodel} model,
many of which have no obvious generalization to arbitrary metrics.
Moreover, several problems in the domain, such as finding the maximum set of links that can simultaneously transmit
(the {\capacity} problem), have been proved to be computationally hard in an
unconstrained SINR model \cite{GHWW09}.

To facilitate algorithmic analysis under the more realistic \decaymodel{} model, we
introduce a \emph{metricity} parameter $\zeta$ that reflects how well
signal decay resembles a metric space.

In what follows, we assume arbitrary path loss with gain $G_{uv} = 1/f(s_u, r_v)$ for some
function $f$ of pairs of points. Note that the RSS between sender $s_u$ to receiver $r_v$ is $P_u G_{uv}$.

\subsubsection{Definition}
The \emph{metricity} $\zeta(x,y)$ of a given node pair $(x,y)$ in gain matrix $G$ is
defined to be the smallest number satisfying for any mote $z$ with links also in $G$,
\begin{equation}
 f(x, y)^{1/\zeta(x,y)} \le f(x,z)^{1/\zeta(x,y)} + f(z,y)^{1/\zeta(x,y)}\ .
\label{eq:zeta}
\end{equation}
For notational simplicity we use $\zeta_v = \zeta(s_v, r_v)$.

\subsubsection{Definition}
We define $\zeta$ as the maximum value over all $\zeta_{v}$ in $G$ unless specified otherwise.
$\zeta$ is well defined, namely, consider $\zeta =
\zeta_0 := \log_2 (f_{max}/f_{min})$, where $f_{max} = \max_{x,y}
f(x,y)$ and $f_{min} = \min_{x,y} f(x,y)$.  Then, we can see that the
LHS of (\ref{eq:zeta}) is at most $f_{max}^{1/\zeta_0} = 2
f_{min}^{1/\zeta_0}$, while the RHS is at least that value.
In the case of geometric path loss, $\zeta \leq \alpha$, since
$f(x,y) = d(x,y)^\alpha$ and the distance function $d(x,y)$ satisfies
the ordinary triangle inequality.

\subsubsection{Theoretical implications}
The $\zeta$ parameter has the advantage that theoretical results 
in the {\geomodel} can be imported \emph{without significant changes} to the {\decaymodel}.
Specifically, the following is true.
\begin{quote}
   All results that hold for general metrics in the {\geomodel} carry over to the {\decaymodel} with trivial modifications, giving practically identical performance ratios in terms of $\zeta$ as the original result had in terms of $\alpha$.
\end{quote}
Results that were proven for general metric spaces therefore do not depend on
the particular value of the path loss constant $\alpha$ or that the value holds homogeneously. 
The results rely upon the triangle inequality, for which
Eqn.~\ref{eq:zeta} applies equally well in arbitrary gain matrices.

Taking the {\capacity} problem as an example, approximation results
carry over for numerous cases: fixed power \cite{SODA11}; arbitrary
power control \cite{KesselheimSODA11,KesselheimESA12}; distributed
setting based on regret minimization \cite{infocom11} and under
jamming \cite{dams2014jamming}; online setting \cite{GHKSV13}, as well
as the weighted version with linear power \cite{us:Infocom12}.

In a sibling paper \cite{us:podc14}, we have examined in detail how
exactly the current body of analytic results carries over to the
{\decaymodel}.  With only few exceptions, results that have been
derived specifically for the Euclidean plane also hold in the {\decaymodel}
model, making only an elementary assumption about the convergence of interference:
that the collective interference of uniformly distributed nodes does not tend to
infinity.

A consequence of these translations between models is that previous theoretical work
in the {\geomodel} model is in fact highly robust to spatial signal variability.

\if
However, results that are tied to the planar setting (or doubling metrics), such as
various distributed algorithms for broadcasting \cite{} and connectivity \cite{}, 
We can generalize this, as stated in the introduction, to claim that any natural SINR algorithm that works in arbitrary metric space, works equally well in the decay model. 
By a natural algorithm we mean one that does not depend on the exact value of $\alpha$ or require it to be identical everywhere.
\fi

\subsection{Experimental evaluation}
\label{sec:zeta}

\subsubsection{Method and evaluation}
Using values obtained for the \RSS matrices in both testbeds, shown for \basement{} in Fig.~\ref{fig_gains},
we evaluate the minimal value of $\zeta_v$, (as defined in Eqn.~\ref{eq:zeta}) for every directed node pair within communication range.
The cumulative distribution function (CDF)
of the resulting $\zeta_v$ values is shown for two frequencies (2.40~\giga\hertz{} and 2.44~\giga\hertz{)} in Fig.~\ref{fig_zeta}.
In both testbeds the values for $\zeta_v$ range up to 12. 
However, note that the metricity values in \basement{} are generally a bit larger than in \classroom{}. 
The discrepancy is to be expected, since the challenging \basement{} environment both contains longer links with more variable signal strength as
well as more variable signal attenuation due to obstacles.

In both testbeds, we find that a small fraction of the links have comparatively higher values of $\zeta_v$,
as seen by the long tapering at the top in Fig.~\ref{fig_zeta}, which in turn drives the value for $\zeta$ to a relatively high number.
To quantify, we find that the $95^{\text{th}}$ and $99^{\text{th}}$ percentile of the $\zeta_v$ value of links are significantly smaller than the global maximum $\zeta$.
In \basement{} some of the weakest links are only able to communicate at particular frequencies.
These links correspond to pairs that have weak signals, and limited or no line-of-sight. 
In \classroom{}, while all pairs have line-of-sight and can communicate with one another, the large $\zeta$ value is an artifact of relatively high $\zeta_v$
values for only a handful of links. We investigate the effect further in Section \ref{sec:multichannel}.

\subsubsection{Factors affecting $\zeta$} 
One of the drivers behind the definition of $\zeta$ is to measure the ``complexity'' of the environment.
To be more precise, we investigated the impact of obstacles in the environment, the number of nodes of a network and the distances between nodes
in order to better pinpoint what features most influence the $\zeta$ value.

To examine the impact of different environmental characteristics, we divide the nodes in \basement{} into three subsets of 20 nodes in Fig.~\ref{fig_zeta_groups}.
The set of motes with the lowest values of $\zeta_v$ has the fewest number of obstacles (\textit{low-obstacle}), 
whereas the set with the higher $\zeta_v$ values (\textit{high-obstacle}) has a variety of barriers, such as electric cables suspended in the ceiling and greater distances between nodes.
The \textit{medium-obstacle} group has an average number of barriers while also having the most condensed topology.
We calculated the values for $\zeta_v$ on the different links in the induced gain matrix for each set to obtain Fig.~\ref{fig_zeta_groups}.
The figure suggests that the complexity of the environment, in the form of physical obstacles, might have a significant impact on the value of $\zeta$.

We further examine the impact of the size of the set of motes on $\zeta$.
Using the \RSS{} matrices for \basement{} and \classroom{}, 
we calculate $\zeta$ for randomly generated subsets of different sizes.
Fig.~\ref{fig_zeta_size} shows metricity as a function of set size.
The thick lines represent $\zeta$ as the average $95^{\text{th}}$ percentile of the values for $\zeta_v$, which is relatively stable
with increased sizes for sets of five or more nodes.
The corresponding thin lines, which represent the average $\zeta$ as the highest value for $\zeta_v$ in each testbed,
demonstrate that the global maximum $\zeta$ continues to grow with the set size.
The increase is to be expected, as larger sets are more likely to include the links with the highest $\zeta_v$ values.

We also looked into the relationship between node distances and the $\zeta$ value of a set.
We took different subsets of nodes that have roughly the same distances to one another, as shown in Fig.~\ref{fig_zeta_distance}.
The figure suggests that the distances between nodes are poorly correlated with the value of $\zeta$.


\begin{figure}[t!]
\centering%
\subfloat[Metricity ($\zeta_v$) in \basement{}]{%
\includegraphics[width=0.5\columnwidth]{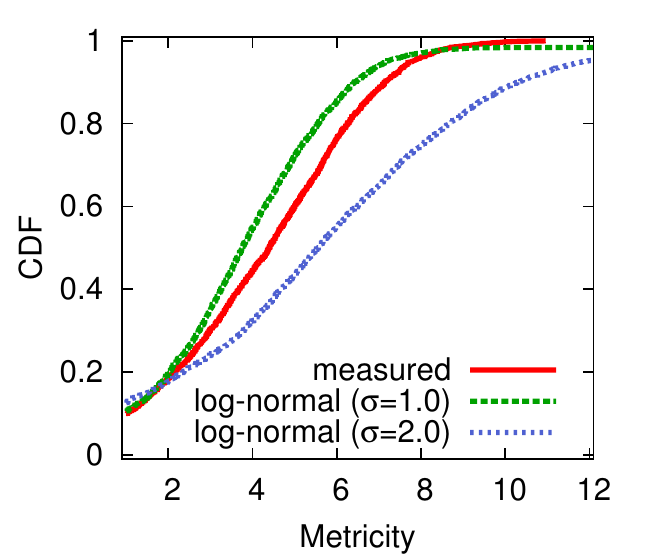}%
\label{fig_lognormal_cdf}%
}%
\subfloat[Metricity ($\zeta$) and standard deviation ($\sigma$)]{%
\includegraphics[width=0.5\columnwidth]{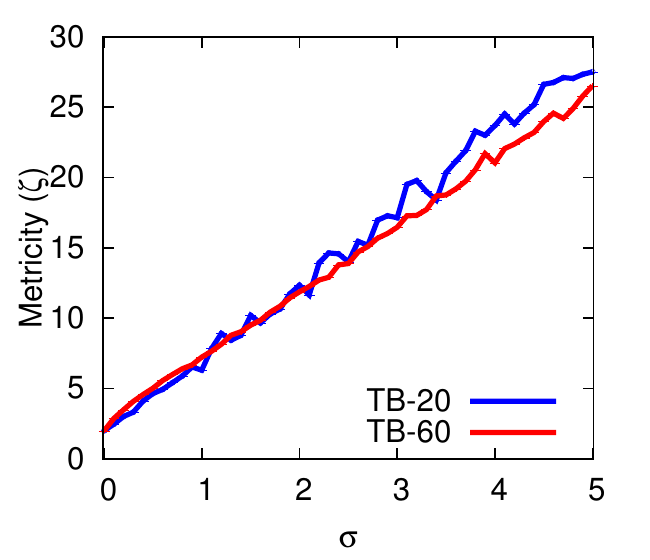}%
\label{fig_lognormal_variance}%
}%
\label{fig_lognormal}%
\caption{
\textbf{(a)} \textbf{Metricity ($\boldsymbol\zeta_v$) with log-normal shadowing in \basement{}}.
CDF of comparison on the metricity $\zeta_v$ using log-normal shadowing on distances in \basement{} averaged over 10 instances compared to actual measurements.
\textbf{(b)} \textbf{Metricity ($\boldsymbol\zeta$) with log-normal shadowing vs.~standard deviation}.
Metricity $\zeta$ (as the $95^{\text{th}}$ percentile on $\zeta_v$) using log-normal shadowing on distances in both testbeds with growing standard deviation $\sigma$.
}
\end{figure}


\subsection{Log-normal shadowing and metricity}
One of the most commonly used stochastic extensions of geometric path loss to address observed variability in signal propagation
is log-normal shadowing.
According to this model, signal decay follows the
geometric model $d^{\alpha}$, but with a multiplicative exponentially distributed
factor:
\begin{equation*}
    f(s_u, r_v) = d(s_u,r_v)^\alpha \cdot e^X,
\end{equation*}
where $X$ is a normally distributed random variable with zero mean.

We note that log-normal shadowing is an approach to introduce
non-geometric properties into gain matrices. Since metricity is
a measure of such discrepancy, it might be instructive to 
calculate metricity on instances generated with log-normal shadowing.

We used the topology of the TB-60 testbed and generated
log-normal distributions with standard deviations $\sigma=1.0$ and $2.0$, and
computed the CDF of the metricity of the resulting gain
matrices. Averages over 10 instances are shown in Fig.~\ref{fig_lognormal_cdf}, interposed
with the actual $\zeta_v$ measurements. We note a similarity 
between the measured values and the log-normal values with $\sigma=1.0$,
although the tail is measurably heavier in the latter.

We also examined how global metricity $\zeta$ grows with the standard deviation
$\sigma$ in the generated log-normal distributed instances in both
testbeds. We show the results on Fig.~\ref{fig_lognormal_variance}, where we display the
$95^{\text{th}}$ percentile of the distribution of $\zeta_v$ values.
The plot exhibits a linear relationship, as shown in Fig.~\ref{fig_lognormal_variance},
or roughly $\zeta = 2.5 + 4.9 \sigma$. 
The agreement accords with the expected signal variations 
produced by a multiple of an exponentially distributed random variable.

%% file: multichanneltheory.tex
\section{Finessing Multi-path Fading with Multiple Channels}
\label{sec:multichannel}

In our experiments, high metricity values were primarily caused by a handful of links.
In particular, experiments in the simple environment of the \classroom{} testbed exhibited a higher value of $\zeta$ than we suspected.
One potential source of the complexity may be due to adverse signal reflection.

When signals travel along different paths, the superposition of the
different signals produces patterns of signal cancellation and
amplification known as multi-path fading
(\cite[Sec.~2]{Goldsmith}). This effect is particularly pronounced and
systematic in simpler settings. The interference pattern will necessarily shift
with frequency.  Hence, the influence of multi-path fading on signal
reception between a pair of points is likely to vary greatly with
the chosen channel. 

We propose to tackle the problem of signal cancellation and
destructive interference by supplying algorithms with multiple channels (frequencies) from which to choose.
First, we propose a variation of
the metricity parameter $\zeta$ and evaluate the difference from the
experimental data from our testbeds. As a case study, we then formulate a
multi-channel version of the link {\capacity} problem and obtain
worst-case approximation results.

\subsection{Experimental Evaluation}

\begin{figure}[t]
\centering%
\includegraphics[width=0.7\columnwidth]{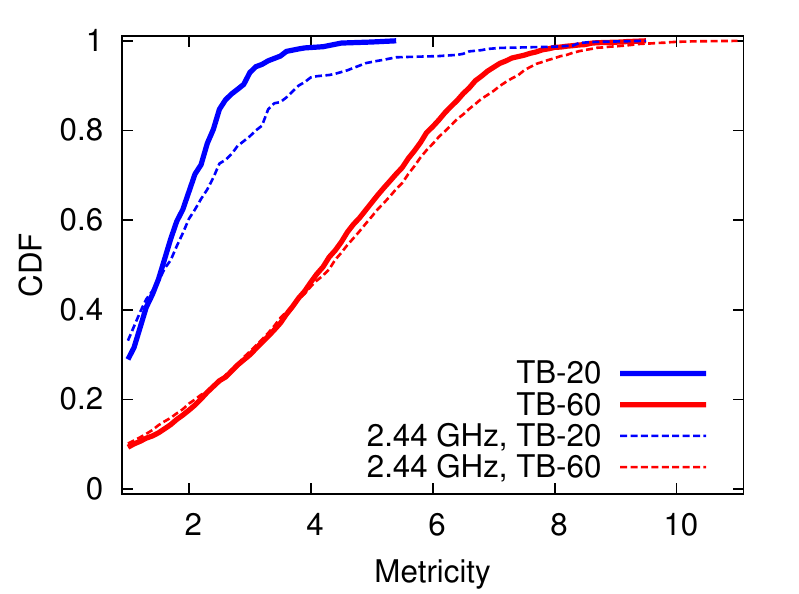}%
\caption{
\textbf{Metricity ($\boldsymbol\zeta_v$) in \classroom{} and \basement{}}.
CDF comparison of computed $\zeta_v$ values for RSS matrices on both testbeds.
The thick lines represent values computed for the $\widehat{\text{RSS}}$ matrices.
The dashed lines show $\zeta$ values for the 2.44~\giga{}\hertz{} frequency, which is representative of other frequencies in the respective testbeds.
}
\label{fig_zeta_multi}%

\end{figure} 
We performed additional experiments to obtain \RSS matrices $\RSS^f$
for 8 different frequencies $f$ ranging from 2.40~\giga\hertz{} to \\
2.48~\giga\hertz{} or wavelengths between 7 and 7.15 cm. These frequencies mean
that multi-path alignment can shift from fully destructive to fully
constructive interference when the difference in the path lengths is at least 1.4 m.

We calculated $\zeta_v$ values separately for each $\RSS^f$.
Although they can vary significantly on a per link basis,
we found the differences in the overall distributions to be insignificant.
Fig.~\ref{fig_zeta_multi} shows $f=2.40~\giga\hertz{}$ --- other frequencies had similar distributions.

To factor out frequency-dependent fading,
we computed for each node pair $(s_v, r_v)$ the median $\widehat{\text{RSS}_v}$ of the eight $\RSS^f_v$ values 
ranging over the different frequencies.
We observe that $\zeta_v$ values of the matrix $\widehat{\text{RSS}}$ are
significantly lower,
as the thick lines in Fig.~\ref{fig_zeta_multi} indicate.

The frequency dependency is more apparent in \classroom{}, which suggests that signal reflection plays a relatively large role in that environment. 
The increased reliance on a particular frequency can be explained by the regular grid structure and condensed setting of \classroom{} (Fig.~\ref{fig_testbed}), which makes the testbed a good candidate to observe (frequency dependent) multi-path fading phenomena. %
However, the links in \basement{} are on average longer and thus reflection plays a smaller role in signal attenuation.
Furthermore, the greater number of obstacles in \basement{} may also explain the decreased dependency on frequency.

The observation that the channels have different fading properties brought us to introduce a new version of the \capacity{} problem to incorporate different frequencies.



\subsection{Link capacity with multiple frequencies}
The empirical indications -- that having a
choice of channels use results in smaller values of $\zeta_v$,
and thus better approximation factors -- motivate us to generalize
the {\capacity} problem. Namely, in {\multicapacity}, we assign links to a set of frequencies,
but each link is only \emph{eligible} to use a subset of the frequencies.
As before, we want to assign as many links as possible with the constraint
that those assigned to a given frequency form a feasible set.

This formulation considers links that experience significant
frequ\-ency-dependent fading as not usable in that frequency. It does
not take into consideration the possible decrease in
\emph{interference} due to such fading. One reason is that such fading
is too unpredictable to expect any algorithm to utilize that to obtain
better solutions than otherwise, and thus it is also not fair to
compare with such a strong adversary.  The other reason is that with
arbitrary fading patterns, we are back in the \emph{abstract SINR}
model, for which very strong inapproximability results hold
\cite{GHWW09}.

The classic {\capacity} problem is to find a maximum subset $S \subseteq L$ of a given set $L$ of links that can successfully transmit simultaneously.
We modify the {\capacity} problem to fit our observations on the use of multiple frequencies:
\smallskip

\multicapacity  
\begin{compactdesc}
\item[\emph{Given:}] A set $L$ of $n$ links, and $k$ subsets $L_1, L_2, \ldots, L_k \subseteq L$.
\item[\emph{Find:}] Sets $S_1, S_2, \ldots, S_k$ with $S_i \subseteq L_i$ and $S_i$ feasible, for $i = 1,2, \ldots, k$.
\item[\emph{Maximize:}] $|S_1 \cup S_2, \ldots \cup S_k|$.
\end{compactdesc}
Here $L_i$ represents the links that are eligible for frequency $i$ and $S_i$ those scheduled for that frequency.
\smallskip

\subsubsection{Additional definitions} 
To simplify notation we write $f_{uv} = f(s_u, r_v)$ and $f_v = f_{vv}$.
We assume a total order $\prec$ on the links, where $\ell_v
\prec \ell_w$ implies that $f_v \le f_w$. We use the shorthand
notation $\ell_v \prec L$ to denote that $\ell_v \prec \ell_u$ for all
links $\ell_u$ in $L$.  A power assignment \cal{P} is \emph{decay
  monotone} if $P_v \le P_w$ whenever $\ell_v \prec \ell_w$,
\emph{reception monotone} if $\frac{P_w}{f_w} \le \frac{P_v}{f_v}$
whenever $\ell_v \prec \ell_w$, and simply \emph{monotone} if both
properties hold.\footnote{This corresponds to \emph{length monotone}
  and \emph{sublinear} power assignments in {\geomodel} \cite{KV10}.} This
captures the main power strategies, including uniform and linear
power.

We modify the notion of \emph{affectance} \cite{GHWW09,HW09,KV10}:
The affectance $a^{\cal{P}}_w(v)$ of link $\ell_w$ on link $\ell_v$ under power assignment $\cal{P}$  is the interference of $\ell_w$ on $\ell_v$ normalized to the signal strength (power received) of $\ell_v$, or
\begin{equation}
a_w(v) = \min \left(1, c_v \frac{P_w G_{wv}}{P_v G_{vv}}\right) = \min \left(1, \frac{P_w}{P_v} \frac{f_v}{f_{wv}}\right)\ ,
\label{eqn:aff}
\end{equation}
where $c_v= \frac{\beta}{1-\beta N/(P_v G_{vv})} > \beta$ is a constant depending only on universal constants and the signal strength $G_{vv}$ of $\ell_v$, indicating the extent to which the ambient noise affects the transmission. 
We drop $\cal{P}$ when clear from context.
Furthermore let $a_v(v) = 0$. For a set $S$ of links and link $\ell_v$, let $a_v(S) = \sum_{\ell_w \in S} a_v(w)$ and $a_S(v) = \sum_{\ell_w \in S} a_w(v)$.
Assuming $S$ contains more than two links we can rewrite Eqn.~\ref{eqn:sinr} as $a_S(v) \leq 1$ and this is the form we will use.
Observe that affectance is additive and thus $a_S(v) = a_{S_1}(v) + a_{S_2}(v)$ for any partition ($S_1, S_2$) of $S$.

We define a weight function $W_+(v,w) = a_v(w) + a_w(v)$, when $\ell_v \prec \ell_w$ and $W_+(v,w)=0$, otherwise.
The plus sign is to remind us that weights are from smaller to larger decay links.
Also, $W_+(X,v) = \sum_{\ell_w \in X} W_+(w,v)$, representing the sum of the in- and out-affectances (as in Eqn. \ref{eqn:aff}) of a link $v$ to and from those links in set $X$ that have smaller decay.

A set $S$ of links is \emph{anti-feasible} if $a_v(S) \leq 2$ for
every link $\ell_v \in S$ and \emph{bi-feasible} if both feasible and
anti-feasible \cite{icalp11}.  More generally, for $K \ge 1$, $S$ is
\emph{$K$-feasible} (\emph{$K$-anti-feasible}) if $a_v(S) \le 1/K$
($a_S(v) \le 2/K$), and \emph{$K$-bi-feasible} if both.


\subsubsection{Approximation of \multicapacity}
\label{sec:fixedpower}
We extend a greedy algorithm for {\capacity} \cite{SODA11} and show
that it gives equally good approximation algorithm for {\multicapacity}, even in
{\decaymodel}.
We assume that the links are assigned monotone power.

\newcommand{\algorithmicbreak}{\textbf{break}}
\newcommand{\BREAK}{\STATE \algorithmicbreak}
\begin{algorithm}[h]
\small
\caption{{\multicapacity} in {\decaymodel}}\label{alg:capfixtri}
\begin{algorithmic}
\STATE Let $L$ be a set of links using monotone power $\calP$ and 
$L_1, L_2, \ldots, L_k \subseteq L$ be subsets.
\STATE Set $X_1, X_2, \ldots X_k \leftarrow \emptyset$
\FOR {$\ell_v \in L$ in order of increasing $f_v$ values}
\FOR {$i \leftarrow 1 \ldots k$}
\IF {$W_+(X_i,v) \leq 1/2$} \label{alg:tri1/2}
\STATE $X_i \leftarrow X_i \cup \{\ell_v\}$
\BREAK
\ENDIF
\ENDFOR
\ENDFOR
\FOR{each $X_i$}
\STATE $S_i \leftarrow \{\ell_v \in X_i| a_{X_i}(v) \leq 1\}$
\ENDFOR
\RETURN $(S_1, S_2, \ldots, S_k)$
\end{algorithmic}
\end{algorithm}
Note that the sets returned by Algorithm \ref{alg:capfixtri} are
feasible by construction.

We turn to proving a performance guarantee for the algorithm.
The following key result bounds the affectance of a feasible set to a
(shorter) link outside the set to a constant. A similar but weaker bound
was first introduced by Kesselheim and V\"ocking \cite{KV10}.

\begin{lemma}
  Let $L$ be a $3^\zeta/\beta$-bi-feasible set with monotone power
  assignment $\calP$ and let $\ell_v$ be a link (not necessarily in
  $L$) with $\ell_v \prec L$. Then, $W_+(v,L) = O(1)$.
\label{thm:W(L,v)}
\end{lemma}

We prove Lemma \ref{thm:W(L,v)} by splitting it into two lemmas
bounding in-affectance for links in a feasible set and similarly
bounding out-affectance for links in an anti-feasible set.

\begin{lemma}
  Let $L$ be a $3^\zeta/\beta$-feasible set with monotone power
  assignment $\calP$ and $\ell_v$ be any link with $\ell_v \prec L$.
  Then, $a^\calP_L(v) = O(1)$.
\label{lem:a_L(v)}
\end{lemma}

\begin{proof}
  The basic idea is to identify a ``proxy'' for $\ell_v$ within the
  set $L$.  Namely, we bound the affectance of $L$ on $\ell_v$ in
  terms of the affectance on the ``nearest'' link $\ell_u$ in $L$,
  which is small since $L$ is feasible and contains $\ell_u$.

\tikzstyle{vertex}=[circle,fill=black!25,minimum size=10pt,inner sep=0pt]
\tikzstyle{edge} = [draw,thick,-]
\tikzstyle{arrow} = [draw,thick,->]
\tikzstyle{line} = [draw, thick, -]
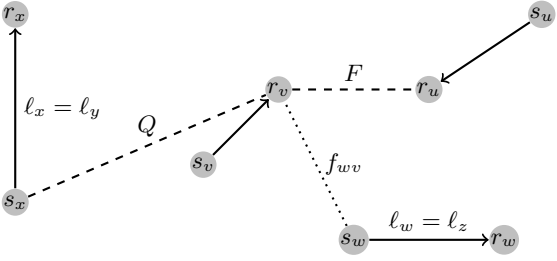
\begin{figure}
\begin{tikzpicture}
\foreach \pos/\name in {{(4.5,2)/s_v},{(5.5,3)/r_v},{(9,4)/s_u},{(7.5,3)/r_u}, {(6.5,1)/s_w},{(8.5,1)/r_w},{(2,1.5)/s_x},{(2, 4)/r_x}}
	\node[vertex](\name) at \pos{$\name$};
\foreach \src/\dest in {s_v/r_v, s_u/r_u} {
	\path[arrow](\src) -- (\dest);}
\path[line,dashed](r_v)--(r_u) node[above, midway]{$F$};
\path[line,dashed](s_x)--(r_v) node[above, midway]{$Q$};
\path[line,dotted](s_w)--(r_v) node[right,midway]{$f_{wv}$};
\path[arrow](s_w)--(r_w) node[above,midway]{$\ell_w = \ell_z$};
\path[arrow](s_x)--(r_x) node[right, midway]{$\ell_x = \ell_y$};
\end{tikzpicture}
\caption{We show that for each link $\ell_x$ it holds that $Q \geq \frac{F}{2^\zeta}$.}\label{fig:structuralvisualization}
\end{figure}

Formally, consider the link $\ell_u = (s_u, r_u) \in L$ such that $F
:=f(r_v, r_u)$ is minimum and link $\ell_w = (s_w, r_w) \in S$ such
that $f_{wv} = f(s_w, r_v)$ is minimum (possibly $\ell_u = \ell_w$).
Let $\ell_x$ be an arbitrary link in $L$ and define $Q = f_{xv}$.  See
Fig.~\ref{fig:structuralvisualization}.

We first show that 
\begin{equation}\label{eq:dxvtri}
F \le 2^\zeta Q \ .
\end{equation}
Let $\ell_y, \ell_z$ be renamings of the links $\ell_x, \ell_w$ such
that $f_y = \max(f_x,f_w)$ and $f_z = \min(f_x,f_w)$.

By definition of $\ell_u$ and $\ell_w$, it holds that
$\max(f_{yv}, f_{zv}) \le Q$.
Thus, using the weak triangular inequality,
\begin{equation}
f(s_y,s_z)^{1/\zeta} \le f_{yv}^{1/\zeta} + f_{zv}^{1/\zeta} \le 2 Q^{1/\zeta} \ .
\label{eq:fswsx}
\end{equation}
Using Eqn. \ref{eq:fswsx} and that $f_z \le f_y$, it holds that
\begin{equation}
 f_{yz}^{1/\zeta} \le f_z^{1/\zeta} + f(s_y,s_z)^{1/\zeta}
      \le f_y^{1/\zeta} + 2 Q^{1/\zeta} \ .
\label{eq:fyz2}
\end{equation}

By the feasibility condition on $L$, $a_y(z) \le \beta/3^\zeta$,
while by definition of affectance and reception monotonicity 
(i.e., $P_y/f_y \le P_z/f_z$),
\[ a_y(z) = c_z \frac{P_y}{P_z} \frac{f_z}{f_{yz}} \ge \beta \frac{f_y}{f_{yz}}\ . \]
Combining the two bounds on $a_y(z)$, we get that  $3^\zeta \cdot f_y \le f_{yz}$.
That, combined with Eqn. \ref{eq:fyz2} and canceling a $f_y$ factor, gives that
$3^\zeta \le \left(1 + 2(Q/f_y)^{1/\zeta}\right)^\zeta$, which implies that $f_y \le Q$
and further that $f_w \le f_y \le Q$.
Then, by the definitions of $F$, $\ell_w$ and $Q$,
\begin{equation}
F^{1/\zeta} = f(r_v,r_w)^{1/\zeta} \le f_{wv}^{1/\zeta} + f_w^{1/\zeta} 
\le 2 \cdot Q^{1/\zeta}\ ,
\label{eq:ftoq}
\end{equation}
implying Eqn.~\ref{eq:dxvtri}, as desired.

Now, using the weak triangular inequality, the definition of $F$,
and Eqn.~\ref{eq:ftoq}, we get that
\[
f_{xu}^{1/\zeta} \le f_{xv}^{1/\zeta} + f(r_v, r_u)^{1/\zeta}
  \le Q^{1/\zeta} + F^{1/\zeta}
  \le Q^{1/\zeta} + 2 Q^{1/\zeta}\ . \]
Thus,
\begin{equation}
f_{xu}  =  3^\zeta Q = 3^\zeta f_{xv}\ . \label{eq:fxufxv}
\end{equation}

Observe that since $f_v \leq f_u$ and power is monotone, it holds that $c_v \leq c_u$.
Then, using Eqn.~\ref{eq:fxufxv} and the definition of affectance,
\begin{equation*}
a_x(v) = c_v \frac{P_x}{f_{xv}}\frac{f_v}{P_v} 
\leq c_u \frac{3^\zeta P_x}{f_{xu}} \frac{f_u}{P_u} = 3^\zeta a_x(u)\ .
\end{equation*}

Finally, letting $L_w = L \setminus \{\ell_w\}$,
we sum over all links in $L$,
$$
a_L(v) = a_w(v) + a_{L_w}(v) 
\leq 1 + 3^\zeta a_{L_w}(u)
\leq 1 + 3^\zeta \cdot \frac{\beta}{3^{\zeta}} = O(1)\ ,
$$
using the feasibility assumption for the last inequality.
\end{proof}

For anti-feasible sets a similar result holds with a nearly
identical proof, swapping the roles of senders and receivers of the links.

\begin{lemma}\label{lem:a_v(L)}
Let $L$ be a $3^\zeta/\beta$-anti-feasible set with monotone power assignment $\calP$ and let $\ell_v$ be a link with $\ell_v \prec L$. Then, $a^\calP_v(L) = O(1)$. 
\end{lemma}

\begin{proof}
Form the dual links $L^*$, which has a link $l^*_v = (s^*_v,r^*_v)$
for each link $l_v=(s_v, r_v )\in L$ such that $s^*_v =r_v$ and $r^*_v = s_v$.
Clearly the lengths of $l^*_v$ and $l_v$ are the same so $f_{v^*v^*} = f_{vv}$.
Also, $f_{v^*u^*} = f_{uv}$.
Observe that the anti-feasibility assumption on $L$ implies that
$L^*$ is $3^\zeta/\beta$-feasible.
Then, we can follow the proof of Lemma \ref{lem:a_L(v)}, applied to the set $L^*$,
to get that 
\[ f_{ux} = f_{x^*u^*} \le 3^\zeta f_{x^* v^*} = 3^\zeta f_{vx}\ . \]
This implies, using the monotonicity of power, that
\begin{equation*}
a_v(x) = c_x \frac{P_v}{f_{vx}}\frac{f_x}{P_x} 
  \leq c_x \frac{3^\zeta P_u}{f_{ux}} \frac{f_x}{P_x} = 3^\zeta a_u(x)\ .
\end{equation*}
The rest of the proof is identical.
\end{proof}
Combining Lemma \ref{lem:a_L(v)} and \ref{lem:a_v(L)} implies Lemma \ref{thm:W(L,v)}.

Finally, to analyze the performance ratio of Algorithm
\ref{alg:capfixtri}, we will use an adaptation of the following
signal-strengthening lemma from \cite[Prop.~8]{FKRV09} and a lemma
generalizing a popular argument used to show that the size of a subset
of links of another set of links is large.

\begin{lemma}[\cite{FKRV09}]
  Let $L$ be a feasible set and $K \ge 1$ be a value.  Then, there
  exists a $K$-bi-feasible subset of $L$ of size $\Omega(|L|/K)$.
\label{lem:fkrw-ss}
\end{lemma}

%

The approximation that we can prove for {\multicapacity} has actually better dependence on $\zeta$ than what follows for {\capacity} from \cite{SODA11}.
A lower bound of $\Omega(2^{\zeta-o(1)})$ on the approximability of {\capacity} \cite{SODA11,us:podc14} implies that the bound is close to best possible.


\begin{theorem}\label{thm:constantcapfixtri}
Algorithm \ref{alg:capfixtri} yields a $O(3^\zeta)$-approximation for \multicapacity.
\end{theorem}

\begin{proof}
Let $L$ be a set of links and let $L_1, L_2, \ldots, L_k \subseteq L$
be subsets of $L$ where $L_i$ contains the links that are eligible in
frequency $i$.  Let $OPT = OPT_1 \cup OPT_2 \cup \ldots \cup OPT_k$ be
an optimum solution to {\multicapacity} on $L$.  Let $K =3^\zeta/\beta$.  
By Lemma~\ref{lem:fkrw-ss}, there is a $K$-bi-feasible subset $OPT_i'$
in $OPT_i$ of size $\Omega(|OPT_i|/K)$, for each $i \in \{1, \ldots,
k\}$.  Let $OPT' = OPT_1' \cup OPT'_2 \cup \ldots \cup OPT'_k$.

Let $S = S_1 \cup S_2 \cup \ldots S_k$ and $X = X_1 \cup X_2 \cup
\ldots X_k$ be the sets computed by Algorithm \ref{alg:capfixtri} on
input $L$.  We first bound $|S|$ in terms of $|X|$ and then $|X|$ in
terms of $|OPT'|$.  To bound $|S|$ to $|X|$ we bound $|S_i|$ to
$|X_i|$ for every $i \in \{1 \ldots k\}$.  Note that by the
construction of $X_i$, $a_{X_i}(X_i) = W_+(X_i,X_i) \le |X_i|/2\ , $
and thus the average in-affectance of links in $X_i$ is at most $1/2$.
Since each $S_i$ consists of the links in $X_i$ of affectance at most $1/2$,
by Markov's inequality, $|S_i| \geq |X_i|/2$.

By the definition of the algorithm, $W_+(X_i,\ell_w) > 1/2$, $\forall
\ell_w \in OPT' \setminus X, X_i \in X$.  Summing over all links
$\ell_w$ in $ OPT' \setminus X$, we get that $W_+(X_i, OPT'\setminus
X) > |OPT'\setminus X|/2$.  Furthermore, since $OPT'$ contains $k$
$K$-feasible sets, it follows by Lemma \ref{thm:W(L,v)} that
$W_+(\ell_v, OPT') = O(k)$, for each $\ell_v \in X$.  Summing over all
links in $X_i$, we get that $W_+(X_i,OPT) = O(k |X_i|)$.  Combining
yields that for any set $X_i$ we have $|OPT'\setminus X|/2 < W_+(X_i,
OPT'\setminus X_i) \in O(k |X_i|)$, giving that $|X_i| =
\Omega(|OPT'\setminus X|/k)$.

Summing over $i$ then gives $|X| = \sum_i |X_i| = \Omega(|OPT' \setminus X|)$.
Thus, the solution output by the algorithm satisfies
$|S| \ge |X|/2  = \Omega(|OPT'|) = \Omega(|OPT|/K) = \Omega(|OPT|/3^\zeta)$.
\end{proof}

In summary, the metricity definition implies that a large range of
algorithmic results from {\geomodel} carries over without change.
Thus, {\decaymodel} has both the desired generality and amenability to
algorithmic analysis.  We also extend known results on {\capacity} to
handle frequency-sensitive links, and improve the dependence of the
approximation on $\zeta$ along the way.


%% file: related.tex
\section{Related Work}
\label{sec:related}



Numerous experimental results have indicated that simplistic range-based models of wireless reception
are insufficient \cite{ganesan2002,Zhao2003,kotz2004experimental,aguayo2004link,padhye2005estimation,Zhou2006,zamalloa2007}. Besides 
directionality and asymmetry, signal strength is not well predicted by distance.
Interference patterns are also insufficiently explained by pairwise relationships,
suggesting the need for additive interference models, both experimentally \cite{kotz2004experimental,Moscibroda2006Protocol,MaheshwariJD2008} 
and analytically \cite{Moscibroda2006Protocol,MoWa06,iyer2009right}. 

The weakness of the known prescriptive models for interference and packet reception
has led experimentalists to form models based on measurements.
Son, Krishnamachari and Heidemann \cite{son2006} sho\-wed that 
the SINR formula, using separately measured RSS values,
is the main factor in predicting PRR.
They found PRR to be dependent on the number of interferers, which was 
not supported in later studies \cite{MaheshwariJD2008,chen2010} and attributed to 
hardware variability or the quality of the CC1000 radios used.
Reis et al.\ \cite{reis2006} independently proposed a similar approach on a 802.11 platform.
They found substantial variability across nodes, and that similarity across time was sufficient 
over moderate time scales of minutes to hours, but that prediction accuracy degrades over longer periods.

Maheshwari, Jain and Das \cite{MaheshwariJD2008} compared different
models of interference using two testbeds with variations in hardware,
power level, and indoor/outdoor. They concluded that the physical model
gives best accuracy, albeit less than perfect.
In their followup workshop paper \cite{MaheshwariJD2008a}, they focus
on the relationship of \emph{joint interference} (SINR with multiple
interferers) to PRR.  They gave strong evidence that the basic formula
works, and verify the \emph{additivity} of the SINR model.

Chen and Terzis \cite{chen2010} proposed a method for calibrating RSSI readings to combine interference measurements from different mo\-tes. 
They found that Tmote Sky motes consistently report RSSI values inaccurately, even reporting non-injective relationships. By aligning measurements from different motes, they obtained much better SINR vs. PRR relationship, reducing the width of the intermediate range significantly. 
They suggested that this may explain much of the imperfect relationship observed in \cite{MaheshwariJD2008}.

Measurement-based approaches have also been proposed in the context of 802.11 \cite{gummadi2007,qiu2007,sevani2012sir}, 
where carrier sense and control packets complicate the picture.
Recent efforts have focused on reducing the required measurements
by deducing interference using, e.g., linear algebra \cite{qi2013} or regression \cite{huang2011}.
Boano et al.~\cite{boano2010} also studied the impact of external
interference on sensor-network MAC protocols, and identified
mechanisms to improve their robustness. 


%
The engineering literature has introduced various extensions in order to capture reality more faithfully.
In the \emph{two-ray} model \cite{Goldsmith}, 
which captures reflection off the ground, 
signal decays in the near-term as a certain polynomial, but as a higher degree polynomial further away.
More generally, the \emph{multi-ray} model
has the signal (in \dBm, log scale) decaying via a piecewise linear model with 
segments of increasing slopes. The function is typically empirically determined.
We note that these models can be captured by {\decaymodel}, with $\zeta$ as the steepest slope.

There are also \emph{empirical models} \cite{Goldsmith}, such as the \emph{Okumura} and \emph{Hata} models, that take the environment into account.
These could also be used to generate a gain matrix.
Also, accurate estimates can be obtained via general \emph{ray tracing} when highly detailed information is available.

Some studies have allowed signal strength to fluctuate from the geometric path loss by up to a constant factor \cite{moscibroda06b,HW09}.
This is of limited help in general, however, since even minor fluctuations of the value of $\alpha$ can 
cause arbitrarily large changes in signal strength \cite{gu2012path}.
%

Various probabilistic models also exist.
On one hand, they are means to prescribe non-geometric components to signal reception, which is useful for simulation studies (which {\decaymodel} cannot provide), but which could be captured more accurately by actual measurements. 
On the other hand, these can also model aspects that are necessarily random, in which case they could complement the deterministic {\decaymodel}. 

